\pdfoutput=1
\documentclass[10pt, a4paper]{article}

\usepackage{amsfonts}
\usepackage{amsmath}
\usepackage{amssymb}
\usepackage{amsthm}
\usepackage{caption}
\usepackage{enumerate}
\usepackage{graphics}
\usepackage{hyperref}
\usepackage{listings}
\usepackage{proof}
\usepackage{subfigure}
\usepackage{syntax}
\usepackage{thmtools}
\usepackage{tikz}
\usepackage{xspace}

\usetikzlibrary{automata,arrows,positioning}

\newcommand{\out}[1]{}

\declaretheorem[style=definition,qed=$\blacktriangle$]{definition}

\theoremstyle{definition}
\newtheorem{lemma}{Lemma}

\newtheorem{theorem}{Theorem}
\newtheorem{corollary}{Corollary}

\newcommand{\decl}{\textit{decl}\xspace}
\newcommand{\stmt}{\textit{stmt}\xspace}
\newcommand{\expr}{\textit{expr}\xspace}

\newcommand{\id}{\textit{id}\xspace}

\renewcommand{\int}{\textit{int}\xspace}

\newcommand{\var}{\textbf{var}\xspace}

\newcommand{\dom}[1]{\ensuremath{\textrm{dom}({#1})}}
\newcommand{\eval}[2]{\ensuremath{\texttt{eval}_{#1}({#2})}}
\newcommand{\FNeval}{\texttt{eval}\xspace}  
\newcommand{\rank}[1]{\ensuremath{\textrm{rank}({#1})}}

\newcommand{\up}[1]{\ensuremath{\lceil{#1}\rceil}}

\newcommand{\sempty}{\texttt{s-empty}\xspace}
\newcommand{\svar}{\texttt{s-var}\xspace}

\newcommand{\okempty}{\texttt{ok-empty}\xspace}
\newcommand{\okstmt}{\texttt{ok-stmt}\xspace}
\newcommand{\okseq}{\texttt{ok-seq}\xspace}
\newcommand{\okprog}{\texttt{ok-prog}\xspace}

\newcommand{\tvar}{\texttt{t-var}\xspace}
\newcommand{\tparen}{\texttt{t-paren}\xspace}
\newcommand{\tprod}{\texttt{t-prod}\xspace}
\newcommand{\ttrans}{\texttt{t-trans}\xspace}
\newcommand{\tcontr}{\texttt{t-contr}\xspace}
\newcommand{\telem}{\texttt{t-elem}\xspace}
\newcommand{\tsmul}{\texttt{t-smul}\xspace}
\newcommand{\tsdiv}{\texttt{t-sdiv}\xspace}

\newcommand{\dempty}{\texttt{d-empty}\xspace}
\newcommand{\dvar}{\texttt{d-var}\xspace}
\newcommand{\evstmt}{\texttt{ev-stmt}\xspace}
\newcommand{\evseq}{\texttt{ev-seq}\xspace}
\newcommand{\evprog}{\texttt{ev-prog}\xspace}
\newcommand{\evprogp}{\texttt{ev-prog'}\xspace}

\newcommand{\dPADempty}{\texttt{d-pad-empty}\xspace}
\newcommand{\dPADvar}{\texttt{d-pad-var}\xspace}
\newcommand{\evPADstmt}{\texttt{ev-pad-stmt}\xspace}
\newcommand{\evPADseq}{\texttt{ev-pad-seq}\xspace}
\newcommand{\evPADprog}{\texttt{ev-pad-prog}\xspace}

\newcommand{\dinputvar}{\texttt{d-in-var}\xspace}

\newcommand{\ok}{\texttt{ok}\xspace}
\newcommand{\e}{\textit{e}\xspace}
\newcommand{\ez}{\textit{e0}\xspace}  
\newcommand{\eo}{\textit{e1}\xspace}  
\newcommand{\x}{\textrm{x}\xspace}
\newcommand{\y}{\textrm{y}\xspace}
\newcommand{\s}{\textrm{s}\xspace}
\renewcommand{\u}{\textrm{u}\xspace}
\renewcommand{\v}{\textrm{v}\xspace}

\newcommand{\A}{\textrm{A}\xspace}
\newcommand{\B}{\textrm{B}\xspace}
\newcommand{\C}{\textrm{C}\xspace}
\renewcommand{\i}{\textbf{i}\xspace}
\renewcommand{\j}{\textbf{j}\xspace}
\renewcommand{\k}{\textbf{k}\xspace}
\renewcommand{\t}{\textbf{t}\xspace}
\newcommand{\M}{\texttt{M}\xspace}

\newcommand{\NN}{\ensuremath{\mathbb{N}}}
\newcommand{\QQ}{\ensuremath{\mathbb{Q}}}
\newcommand{\RR}{\ensuremath{\mathbb{R}}}
\newcommand{\VV}{\ensuremath{\mathbb{V}}}

\begin{document}
\title{%
	Modeling of languages for tensor manipulation%
}%
\author{%
	Norman A.~Rink \\
	norman.rink@tu-dresden.de \\
	Technische Universit\"at Dresden \\
	Germany
}%
\date{January 26, 2018}
\maketitle
\begin{abstract}
Numerical applications and, more recently, machine learning applications rely on high-dimensional data that is typically organized into multi-dimensional tensors.
Many existing frameworks, libraries, and domain-specific languages support the development of efficient code for manipulating tensors and tensor expressions.
However, such frameworks and languages that are used in practice often lack formal specifications.
The present report formally defines a model language for expressing tensor operations.
The model language is simple and yet general enough so that it captures the fundamental tensor operations common to most existing languages and frameworks. 
It is shown that the given formal semantics are sensible, in the sense that well-typed programs in the model language execute correctly, without error.
Moreover, an alternative implementation of the model language is formally defined.
The alternative implementation introduces padding into the storage of tensors, which may benefit performance on modern hardware platforms.
Based on their formal definitions, the original implementation of the model language and the implementation with padding are proven equivalent.
Finally, some possible extensions of the presented model language for tensor manipulation are discussed.
\end{abstract}
\section{Introduction}%
\label{sec:introduction}%
High-dimensional data structures are ubiquitous.
For decades, numerical applications have relied heavily on storing and processing data that is organized in multi-dimensional arrays.
More recently, machine learning applications have also started to make extensive use of large vectors, matrices, and arrays of higher dimensions.
Multi-dimensional arrays are not merely an organization of data in memory; %
they also support a number of natural operations.
Matrix multiplication is perhaps the most widely used such operation, but other operations from linear algebra are important too.
In the context of software engineering, multi-dimensional arrays that are accompanied by natural operations are casually referred to as \emph{tensors}, to emphasize the presence of algebraic structure beyond the mere organization of data into arrays.

Over the years, many libraries and frameworks have appeared that support the development of efficient codes for manipulating tensors, e.g.~\cite{ Baumgartner:2005:TCE,numpy:2017,Bergstra:2010:Theano,Abadi:2016:TensorFlow,itensor:2017,Springer:2017:TTC,Susungi:2017:GPCE,Kjolstad:2017:TACO,Wei:2017:DLVM,Rink:2018:RWDSL}.
Frameworks often define domain-specific languages (DSLs) that facilitate that code can be written at a high level of abstraction while guaranteeing good performance when code is executed.%
\footnote{%
	Sometimes the DSL appears in the guise of a library's API. %
	In this case, the definition of the DSL is implicit in the API specification. %
} %
The formal semantics of these DSLs are often left implicit or are introduced in an ad-hoc fashion.
Since the mathematical operations that underlie tensor manipulation are very well understood, it is perhaps unsurprising that formalizing the semantics is not always considered a high-priority task in the design and implementation of such DSLs.
However, the absence of a formal specification of language semantics makes it difficult to verify correct behavior of implementations.
This is particularly problematic when a language implementation applies many code transformations to optimize code for fast, potentially parallel~\cite{Steuwer:2016:beyond:auto-tuning,Steuwer:2017:Lift} execution.
\out{%
In order to verify that each transformation preserves the language semantics, the semantics have to be defined in the first place.
}

As a step towards overcoming the problems introduced by the absence of formal definitions, the present report specifies formal semantics for tensor manipulation.
To this end, we define a model language that captures essential tensor operations (Section~\ref{sec:background}), %
namely: %
element-wise operations, e.g.~addition or subtraction; %
transposition of dimensions; %
and tensor contraction, which is the natural generalization of matrix multiplication to tensors of more than two dimensions.
By supporting these operations, our model language can express the same tensor manipulations that are commonly supported by libraries and frameworks~\cite{Baumgartner:2005:TCE,numpy:2017,Bergstra:2010:Theano,Abadi:2016:TensorFlow,itensor:2017,Susungi:2017:GPCE,Rink:2018:RWDSL}.
Therefore, our analysis and results should carry over to these frameworks.
Specifically, we formally establish a reference implementation of the model language (Section~\ref{sec:definition}) that we also prove to be safe, in the sense that there are no out-of-bounds accesses to multi-dimensional arrays (Section~\ref{sec:progress:safety}).
We then introduce an alternative implementation that uses padding in the storage layout of tensors (Section~\ref{sec:padding}).
Padding can achieve alignment of tensors in memory and may generally lead to better performance on modern hardware platforms.
Thanks to the formal definition of language semantics, the implementation with padding can be proven equivalent to the reference implementation (Section~\ref{sec:padded:simulation}).
Finally, a number of straightforward extensions of the model language are discussed (Section~\ref{sec:extensions}), and the present report is summarized (Section~\ref{sec:summary}).
\section{Background}%
\label{sec:background}%
Numerical applications rely heavily on linear algebra, %
and hence matrix multiplication is the key operation in many of these applications.
Depending on the specific application, matrix multiplication may appear in different variants and guises.
Nonetheless, the multiplication of matrices $A$ and $B$ is fundamentally defined as
\begin{align}
	(A B)_{ij} = \sum_{k=1}^{d} A_{ik} \cdot B_{kj}.
	\label{eq:matrix:multiplication}
\end{align}
Of course, this only makes sense if the extent of the second dimension of $A$ equals the extent of the first dimension of $B$; %
both are denoted as $d$ in Equation~\eqref{eq:matrix:multiplication}.
For multi-dimensional tensors $u$ and $v$, \emph{tensor contraction} generalizes matrix multiplication, %
and applications that work with high-dimensional data may rely more on tensor contraction than matrix multiplication.
The contraction of the $m$-th dimension of $u$ with the $n$-th dimension of $v$ is defined as
\begin{align}
	\sum_{l=1}^{d} u_{i_1\dots i_{m-1} l\, i_{m+1}\dots i_{k_1}}
		\!\cdot v_{j_1\dots j_{n-1} l\, j_{n+1}\dots j_{k_2}} .
	\label{eq:contraction}
\end{align}
As for matrix multiplication, the contraction only makes sense if the $m$-th dimension of $u$ and the $n$-th dimension of $v$ have the same extent $d$.
This clearly hints at a need for typing of tensors and for type checking of operations~\cite{Eaton:2006:Haskell,Chen:2017:typesafe:abstractions}.

Henceforth, we refer to the number of dimensions of a tensor as the tensor's \emph{rank}.
Thus, in the previous example, $k_1 = \rank{u}$ and $k_2 = \rank{v}$.
It follows that the contraction of $u$ and $v$ has rank $k_1+k_2-2$, regardless of which pair of dimensions is contracted over.
This implies that the type of a tensor must contain more fine-grained information than the tensor's rank.

It is of course also possible to contract over a pair of dimensions of the same tensor.
This kind of contraction naturally occurs in traces of matrix products, which are related to the Euclidean norm~\cite{Froberg:1965} or Frobenius norm~\cite{Strang:2010}:
\begin{align}
	\textit{tr}(A B) =
		\sum_{k_1=1}^{d_1} \left( \sum_{k_2=1}^{d_2} A_{k_1 k_2} B_{k_2 k_1} \right) ,
	\label{eq:trace:of:product}
\end{align}
where parentheses have been added to stress that the result of the inner contraction is a tensor (of rank $2$) in its own right.

While contraction reduces the rank, the \emph{outer product} of tensors $u$ and $v$, denoted $u \otimes v$, constructs a tensor of rank $k_1\! + k_2$:
\begin{align}
	(u \otimes v)_{i_1\dots i_{k_1} j_1\dots j_{k_2}} = 
		u_{i_1\dots i_{k_1}} \!\cdot v_{j_1\dots j_{k_2}} .
	\label{eq:outer:product}
\end{align}
In practical applications, however, rather than forming higher-rank tensors by using the outer product, high-dimensional data is often decomposed into a product of lower-rank tensors~\cite{Stoudenmire:2016}.
The product structure can reduce the complexity of numerical operations, which then benefits performance.
Moreover, this decomposition of data into outer products is often the reason that tensors are introduced into an application in the first place~\cite{Lynch:1964:2,Lynch:1964:1}.
Note that the outer product is frequently also referred to as the \emph{tensor product}.

Another operation that generalizes from matrices to tensors of higher rank is \emph{transposition}.
The tensor $v$ is obtained from $u$ by transposition of the $m$-th and $n$-th dimension if
\begin{align}
	v_{i_1\dots i_m \dots i_n \dots i_k} = u_{i_1\dots i_n \dots i_m \dots i_k} .
	\label{eq:transposition}
\end{align}
Transposition is very important in the implementation of numerical operations since the layout of a data structure in memory can have a big impact on performance~\cite{Springer:2016,Springer:2017:TTC}.

Finally, the usual arithmetic operations generalize to element-wise operations on tensors:
\begin{align}
	(u \odot v)_{i_1\dots i_{k_1}} =
		u_{i_1\dots i_{k_1}} \odot v_{i_1\dots i_{k_2}},
		\quad \odot\in\{+,-,*,/\} .
	\label{eq:element-wise}
\end{align}
For this to make sense, we must have $k_1 = k_2$ and each of the dimensions of $u$ and $v$ must have the same extent.
Once again, this calls for type checking of expressions involving tensors.

Any library, framework, or DSL that aims to support tensor operations will offer an interface or syntax for expressing the operations that have been introduced above.
Beyond merely offering sufficient expressiveness, frameworks and DSLs are typically designed to achieve one or several of the following goals:
\begin{enumerate}[{(1)}]
	\item
	reduction of development time by offering high levels of abstraction, e.g.~\cite{Kirby:2006,Bergstra:2010:Theano,Ragan-Kelley:2013:Halide,Abadi:2016:TensorFlow,Membarth:2016:Hipacc,Perard-Gayot:2017:RaTrace};
	~\\[-6mm]
	\item
	automatic optimization of codes for high performance, e.g.~\cite{Puschel:2004:SPIRAL,Baumgartner:2005:TCE,Ragan-Kelley:2013:Halide,Membarth:2016:Hipacc,Spampinato:2016:LGen,Springer:2017:TTC,Wei:2017:DLVM,Rink:2018:RWDSL};
	~\\[-6mm]
	\item
	ability to re-target codes at different platforms, ideally without losing performance, e.g.~\cite{DeVito:2011:Liszt,Ragan-Kelley:2013:Halide,Steuwer:2015:rewriting,Steuwer:2016:beyond:auto-tuning,Wei:2017:DLVM}.
\end{enumerate}
Achieving these goals typically requires that high-level code undergoes a number of transformations, as is the case in optimizing compilers for general-purpose languages~\cite{Peyton-Jones:1987,Muchnick:2000,Aho:2003:dragonbook}.
To verify formally that transformations do not change the behavior of programs, one first requires a formal definition of program behavior, i.e.~a definition of program execution and its effects.
The present report provides this definition for a model language designed for expressing operations on tensors.
The language, which is introduced and defined formally in the next section, serves as a model for most existing DSLs for tensor manipulation.
\section{The model language}%
\label{sec:definition}%
The previous section has introduced the relevant tensor operations that ought to be supported by any framework or language for tensor manipulation.
Based on this requirement, the present section defines a model language for expressing tensor operations.
In order to focus attention specifically on tensor manipulation, the model language is deliberately kept as minimal as possible.

Another reason for keeping the model language small is to make it easy to see that many existing frameworks and DSLs can be mapped to it.
By establishing a suitable mapping for a given DSL for tensor manipulation, e.g.~\cite{numpy:2017,itensor:2017,Bergstra:2010:Theano,Abadi:2016:TensorFlow,Susungi:2017:GPCE,Wei:2017:DLVM}, one can carry our subsequent analyses over to that DSL.
Thus, one should not think of the model language as a ``silver bullet'' for all problems related to tensors.
It is rather an intermediate language that is intended for the study of formal properties, in particular safety, and correctness of transformations.
\grammarindent0.5in
\begin{figure}
	\begin{grammar}
		<prog> ::= <decl>* <stmt>*
		
		<decl> ::= \var id \textbf{:} \textbf{[}tuple\textbf{]}
		
		<stmt> ::= <id> \textbf{=} <expr> 
		
		<expr> ::= <id> \:|\: \textbf{(} <expr> \textbf{)} \:|\: <expr> ($+$|$-$|$*$|$/$) <expr> \\
		\:|\: <expr> \textbf{\#} <expr> \:|\: <expr> \textbf{.} \textbf{[}<pair>\textbf{]} \:|\: <expr> \textbf{\^{}} \textbf{[}<pair>\textbf{]}
		
		<pair> ::= <int> <int>
		
		<tuple> ::= $\epsilon$ | <int> <tuple>
		
		<int> ::= [0-9][0-9]*
		
		<\id> ::= [a-zA-Z][a-zA-Z0-9]*
	\end{grammar}
	\caption{%
		Syntax of the model language. %
		Non-terminals in angle brackets. %
	}%
	\label{fig:grammar}
\end{figure}%
\subsection{Syntax and examples}
\label{sec:syntax}
A program in the model language consists of a sequence of declarations followed by a sequence of statements, cf.~Figure~\ref{fig:grammar}.
The declaration of a variable starts with the keyword \var, after which the identifier appears, followed by a tuple of integers.
The tuple specifies the extents of the tensor's dimensions and hence the type of the declared variable.
For example, a tensor of rank 3 whose dimensions contain 300, 400, and 500 elements respectively is declared in Figure~\ref{fig:tensor:declaration}.
The degenerate case where a variable is declared to denote a scalar is shown in Figure~\ref{fig:scalar:declaration}.
In this case, the scalar type is indicated by an empty tuple.

\vspace{6mm}
\begin{minipage}{0.4\textwidth}
	\begin{center}
	\var \x\hspace{0.4mm}: \textbf{[}300 400 500\textbf{]}
	\end{center}
	\vspace{-4mm}
	\captionof{figure}{Tensor declaration.}
	\label{fig:tensor:declaration}
\end{minipage}
\hspace{8mm}
\begin{minipage}{0.4\textwidth}
	\begin{center}
	\var \s\hspace{0.4mm}: \textbf{[}\hspace{0.4mm}\textbf{]}
	\end{center}
	\vspace{-4mm}
	\captionof{figure}{Scalar declaration.}
	\label{fig:scalar:declaration}
\end{minipage}
\vspace{6mm}

A statement in the model language assigns a tensor expression to a previously declared variable.
Variables and tensor expressions are written without indices, as in~\cite{numpy:2017,Bergstra:2010:Theano,Abadi:2016:TensorFlow,Rink:2018:RWDSL}.
The dimensions of tensor variables and expressions are implicitly numbered, from left to right, starting from one.

All of the tensor operations introduced in Section~\ref{sec:background} can be expressed in the model language.
The hash operator (\textbf{\#}) denotes the outer product of its operands, and contraction is expressed by a period (\textbf{.})~followed by a pair of integers that specifies the dimensions to be contracted.
The code that expresses the matrix multiplication from Equation~\eqref{eq:matrix:multiplication} is shown in Figure~\ref{fig:contraction}.
First, the outer product of matrices \A and \B is formed with the hash operator, and then the second and third dimension are contracted over. 
In Figure~\ref{fig:contraction}, the resulting matrix product is assigned to the variable \C, which is declared with appropriate dimensions.

\vspace{6mm}
\begin{minipage}[t]{0.38\textwidth}
		\var \A\hspace{0.4mm}: \textbf{[}300 400\textbf{]} \\
		\var \B\hspace{0.4mm}: \textbf{[}400 500\textbf{]} \\
		\var \C\hspace{0.4mm}: \textbf{[}300 500\textbf{]} \\~\\
		\C = (\A \textbf{\#} \B) \textbf{.} \textbf{[}2 3\textbf{]}
	\captionof{figure}{Contraction.}
	\label{fig:contraction}
\end{minipage}
\hspace{8mm}
\begin{minipage}[t]{0.48\textwidth}
		\var \A\hspace{0.4mm}: \textbf{[}300 400\textbf{]} \\
		\var \B\hspace{0.4mm}: \textbf{[}400 300\textbf{]} \\
		\var \s\hspace{1.3mm}: \textbf{[}\hspace{0.4mm}\textbf{]} \\~\\
		\s = ((\A \textbf{\#} \B) 
				\textbf{.} \textbf{[}2 3\textbf{]})
				\textbf{.} \textbf{[}1 2\textbf{]}
	\captionof{figure}{Trace of matrix product.}
	\label{fig:trace:of:product}
\end{minipage}
\vspace{6mm}

Figure~\ref{fig:trace:of:product} demonstrates how multiple contractions can be expressed: %
the trace of the matrix product of \A and \B, cf.~Equation~\eqref{eq:trace:of:product}, is computed.
The result of the first contraction has type \textbf{[}300 300\textbf{]}, and both dimensions are contracted in the second contraction.
Hence, the final result is a scalar and is assigned to the variable \s.
Note that rather than expressing multiple contractions by a sequence of periods as in Figure~\ref{fig:trace:of:product}, one could introduce syntax for writing multiple contractions more compactly.
Such syntax is indeed used in existing frameworks and DSLs for tensors~\cite{numpy:2017,itensor:2017,Baumgartner:2005:TCE,Bergstra:2010:Theano,Abadi:2016:TensorFlow,Susungi:2017:GPCE,Rink:2018:RWDSL}.
Since contractions over multiple pairs of dimensions can always be reduced to a sequence of contractions over single pairs of dimensions, no generality is lost in the model language by restricting the period operator to only one pair of dimensions.

Finally, transposition is expressed by the caret operator (\textbf{\^{}}).
Thus, the code for Equation~\eqref{eq:transposition} is as in Figure~\ref{fig:transposition}.
Again, an arbitrary permutation of a tensor's dimensions can be expressed as a sequence of transpositions.
Hence, the model language's caret operator is completely general, while existing languages will typically include syntax for expressing general permutations more compactly.

\vspace{6mm}
\begin{minipage}[t]{0.4\textwidth}
	\var \u\hspace{0.4mm}: \textbf{[}200 300 400 500 600\textbf{]} \\
	\var \v\hspace{0.4mm}: \textbf{[}200 500 400 300 600\textbf{]} \\~\\
	\v = u \textbf{\^{}} \textbf{[}2 4\textbf{]}
	\captionof{figure}{Transposition.}
	\label{fig:transposition}
\end{minipage}
\vspace{6mm}

The production for \expr in the grammar in Figure~\ref{fig:grammar} does not specify any precedence order for the operations between expressions.
For the purpose of this report, no precedence order is required, and it may be assumed that all operations between expressions associate to the left.
Hence, parentheses must be used to indicate the order in which operations are to be performed.%
\footnote{%
	An alternative convention could be to assign higher precedence to those operators that appear further towards the right in the production for \expr.
	This convention is in agreement with the usual precedence order of the arithmetic operators $+$, $-$, $*$, and $/$.
}
\subsection{Typing and evaluation}
\label{sec:semantics}
In this section we formally define the semantics of the model language.
We also establish some simple properties of model language programs.
Before we can do this, we need to set up a number of preliminary definitions.
\begin{definition}[Multi-index, partial order, projection]
\label{def:multi:index}
Let $k\in\NN_0$.
\begin{enumerate}[{(i)}]
	\item
	A \emph{multi-index} $\i$ is a tuple $\i\in\NN^k$. %
	By convention, $\NN^0 = \{()\}$, i.e.~the set containing only the empty tuple.
	Multi-indices are written in bold face.
	\item
	\label{def:multi:index:order}
	Let $\i = (i_1,\dots,i_k) \in\NN^k$ and $\j = (j_1,\dots,j_k) \in\NN^k$.
	The partial order $\le_k$ is defined by
	\begin{align*}
		\i \le_k \j \:\Leftrightarrow\: \forall_{1\le l\le k}\, i_l \le j_l .
	\end{align*}
	The subscript $k$ is usually omitted, i.e.~instead of $\le_k$ we simply write $\le$.
	\item
	\label{def:multi:index:interval}
	Let $\i, \j, \k \in\NN^k$.
	A pair of multi-indices is denoted as $(\i,\j)$.
	Mildly abusing notation, we write $\k \in (\i,\j)$ iff $\k \nleq \i$ and $\k \le \j$. 
	\item
	The elements of a tensor $x$ can be accessed using multi-indices,
	\begin{align*}
		x_{\i} = x_{i_1\cdots i_k},
	\end{align*}
	where $\i = (i_1,\dots,i_k) \in\NN^k$.
	\item
	\label{def:multi:index:projection}
	For all $l\in\NN$ such that $1\le l\le k$, %
	the $l$-th \emph{projection} $\pi_l$ is defined as the map %
	$\pi_l : \NN^k \to \NN$, $\pi_l(\i) = i_l$, for $\i = (i_1,\dots,i_k) \in \NN^k$.
	\end{enumerate}%
	~\\[-11mm]
\end{definition}
Multi-indices are also used to specify the type of a tensor.
However, our presentation will be clearer if we can distinguish tensor types from general multi-indices.
Hence, the following definition of tensor types.
\begin{definition}[Tensor type, rank]
\label{def:tensor:type}
A \emph{tensor type} $\t$ (or simply \emph{type}) is a multi-index $\t \in\NN^k$, $k\in\NN_0$.
For $\t \in\NN^k$, $k$ is the \emph{rank} of the tensor type $\t$.
\end{definition}
Note that the rank of a tensor may be zero (since $k\in\NN_0$), which accounts for the degenerate case when a tensor is in fact a scalar, as in Figure~\ref{fig:scalar:declaration}.
The extent of any tensor dimension, however, must be positive since the components of a multi-index are taken from $\NN= \{1,2,\dots\}$.

The essential ingredients in typing and evaluating model language programs are two distinct maps.
The first map, denoted $\Gamma$, maps identifiers to tensor types.
This map is used during the typing of model language programs to access the types of declared variables.
Since typing is a static program analysis, we refer to the map $\Gamma$ as the \emph{static context}.

The second map, denoted $\mu$, maps subscripted identifiers to an appropriately chosen value domain, typically $\RR$, $\QQ$, or a machine-representable subset thereof.
Thus, the map $\mu$ stores the values corresponding to elements of tensors.
During program evaluation, the map $\mu$ is typically updated.
Hence, $\mu$ is referred to as the \emph{dynamic store}.
When a model language program has terminated, the state of the store $\mu$ is considered the program's result since $\mu$ then contains the values computed for all the tensors involved in the program.

The next definition summarizes the previous paragraphs.
\begin{definition}[Context, store]
\label{def:context:store}
~\\[-5mm]
\begin{enumerate}[{(i)}]
	\item
	\label{def:context:store:context}
	A \emph{static context} (or simply \emph{context}) is a map $\Gamma$ from identifiers to tensor types of arbitrary rank.
	Formally, for an identifier \x, if $\x\in\dom{\Gamma}$, then
	\begin{align*}
		\Gamma(\x)\in\bigcup_{k=0}^{\infty}\NN^k.
	\end{align*}
	\item
	\label{def:context:store:undefined}
	The symbol $\bullet$ denotes the special value \emph{undefined}.
	\item 
	\label{def:context:store:store}
	A \emph{dynamic store} (or simply \emph{store}) is a map $\mu$ from subscripted identifiers to values, including the undefined value $\bullet$.
	Formally, for an identifier \x and multi-index $\i\in\NN^k$, %
	if $\x_{\i}\in\dom{\mu}$, then
	\begin{align*}
		\mu(\x_{\i}) \in \VV \cup \{\bullet\},
	\end{align*}
	where $\VV$ is the value domain of interest.
\end{enumerate}%
~\\[-11mm]
\end{definition}
We admit $\bullet$ in the image of the store $\mu$ for two reasons.
First, this allows us to reason about variables that are in the domain of $\mu$ but have not been initialized.
Second, by allowing undefined values, we do not need to take extra care in handling division by zero.
Extensions of the usual arithmetic operators to $\VV \cup \{\bullet\}$ are defined below.
\newpage
\begin{definition}[Extended arithmetic]
\label{def:extended:arithmetic}
Let $v\in\VV$ and $w\in\VV\setminus\{0\}$.
Then, %
\begin{enumerate}[{(i)}]
	\item
	$v \odot \bullet = \bullet \odot v = \bullet$, for $\odot\in\{+,-,*\}$,
	\item
	$v / \bullet = \bullet / w = \bullet$, and %
	$v / 0 = \bullet / 0 = \bullet$.
\end{enumerate}
~\\[-11mm]
\end{definition}
Since we allow maps to take the undefined value, we must distinguish carefully between the constant map that takes the value $\bullet$ everywhere and the map $\bot$ that is not defined anywhere.
Hence, we define the map $\bot$ by virtue of its empty domain.
\begin{definition}[Map with empty domain]
\label{def:bottom}
The symbol $\bot$ denotes the map that has an empty domain, i.e.~$\dom{\bot}=\emptyset$.
\end{definition}
\begin{figure*}[h!]
	\begin{minipage}{\textwidth}
		\begin{gather*}
		\infer[\sempty]
		{\emptyset \vdash_{s} \bot}
		{} \\[2mm]
		\infer[\svar]
		{\decl^*\:\:
			\textbf{var}\: \x\: \textbf{:}\: \textbf{[}d_1, \dots, d_k\textbf{]}
			\vdash_{s} \Gamma\{\x \mapsto (d_1,\dots,d_k)\}}
		{\decl^* \vdash_{s} \Gamma
			\quad \x\notin \dom{\Gamma}}
		\end{gather*}
		\vspace{-6mm}
		\captionof{figure}{Static context formation $\vdash_{s}$.}
		\label{fig:typing:context}
	\end{minipage}%
	\\[5mm]
	\begin{minipage}{\textwidth}
		\begin{gather*}
		\infer[\okempty]
		{\Gamma \vdash \emptyset\:\:\texttt{ok}}
		{} \\[2mm]
		\infer[\okstmt]
		{\Gamma \vdash \x\: \textbf{=}\: \e \:\: \texttt{ok}}
		{\x\in\dom{\Gamma} \quad
			\t = \Gamma(\x) \quad
			\Gamma \vdash \e : \t} \\[2mm]
		\infer[\okseq]
		{\Gamma \vdash \stmt^*\:\stmt\:\:\texttt{ok}}
		{\Gamma \vdash \stmt^*\:\:\texttt{ok} \quad
			\Gamma \vdash \stmt\:\:\texttt{ok}} \\[2mm]
		\infer[\okprog]
		{\decl^*\:\stmt^*\:\:\texttt{ok}}
		{\decl^* \vdash_{s} \Gamma \quad
			\Gamma \vdash \stmt^*\:\:\texttt{ok}}
		\end{gather*}
		\vspace{-6mm}
		\captionof{figure}{Statement and program well-formedness.}
		\label{fig:typing:statement}
	\end{minipage}%
	\\[5mm]
	\begin{minipage}{\textwidth}
		\begin{gather*}
		\infer[\tvar]
		{\Gamma \vdash \x : \t}
		{\x\in\dom{\Gamma} \quad
			\t = \Gamma(\x)} \\[2mm]
		\infer[\tparen]
		{\Gamma \vdash \textbf{(} \e\,\textbf{)} : \t}
		{\Gamma \vdash \e : \t} \\[2mm]
		\infer[\tprod]
		{\Gamma \vdash \ez\,\textbf{\#} \eo
			: (d_{01},\dots,d_{0k},d_{11},\dots,d_{1l}) }
		{\Gamma \vdash \ez : (d_{01},\dots,d_{0k}) \quad
			\Gamma \vdash \eo : (d_{11},\dots,d_{1l})} \\[2mm]
		\infer[\ttrans]
		{\Gamma \vdash \e\,\textbf{\^{}}\,\textbf{[}m\, n\textbf{]} : (d_1, \dots, d_n, \dots, d_m, \dots, d_k)}
		{\Gamma \vdash \e : (d_1, \dots, d_m, \dots, d_n, \dots, d_k)} \\[2mm]
		\infer[\tcontr]
		{\Gamma \vdash \e\,\textbf{.}\,\textbf{[}m\, n\textbf{]} : (d_1, \dots, \widehat{d_m}, \dots, \widehat{d_n}, \dots, d_k)}
		{\Gamma \vdash \e : (d_1, \dots, d_m, \dots, d_n, \dots, d_k) \quad
			d_m = d_n} \\[2mm]
		\infer[\telem]
		{\Gamma \vdash \ez \odot\! \eo : \t}
		{\Gamma \vdash \ez : \t \quad
			\Gamma \vdash \eo : \t \quad
			\odot\in\{+, -, *, /\}} \\[2mm]
		\infer[\tsmul]
		{\Gamma \vdash \ez *\! \eo : \t_1}
		{\Gamma \vdash \ez : () \quad
			\Gamma \vdash \eo : \t_1} \\[2mm]
		\infer[\tsdiv]
		{\Gamma \vdash \ez\,/ \eo : \t_0}
		{\Gamma \vdash \ez : \t_0 \quad
			\Gamma \vdash \eo : ()}
		\end{gather*}
		\vspace{-6mm}
		\captionof{figure}{Expression typing.}
		\label{fig:typing:expression}
	\end{minipage}%
\end{figure*}%
\subsubsection{Typing and well-formed programs}
\label{sec:typing}
\out{%
	Now we are finally in a position to discuss the rules for forming valid model language programs.
}
The inference rules in Figures~\ref{fig:typing:context}--\ref{fig:typing:expression} specify what constitutes a correct program in the model language.
The declarations at the beginning of a program are used to form a static context $\Gamma$, cf.~Figure~\ref{fig:typing:context}.
Initially, the context is empty, which is expressed by the appearance of $\bot$ in the rule \sempty.
Subsequently, rule \svar states that every declaration of a variable adds the variable's identifier to the context $\Gamma$, such that $\Gamma$ maps the identifier to its tensor type.
The notation $\Gamma\{\x\mapsto \dots\}$ means that the domain of the map $\Gamma$ is augmented by $\x$.
Note that because of $\x\notin\dom{\Gamma}$ in the premise of rule \svar, variables may not be re-declared in the model language.

Figure~\ref{fig:typing:statement} lists the inference rules for judging whether a program is well-formed.
Rule \okprog says that in order to decide whether a program is well-formed, the static context $\Gamma$ must first be formed from the declarations that appear in the program.
This context is then used to decide whether the statements in the program are well-formed.
By the rule \okempty, an empty sequence of statements is well-formed.
By \okstmt, an assignment statement is well-formed if the type of the variable and the type of the expression on the right-hand side agree.
Note that $\x\in\dom{\Gamma}$ in the premise of rule \okstmt ensures that variables have been declared when expressions are assigned to them.

The rule \okstmt relies on the typing of an expression $\e$ in the context $\Gamma$.
Expression typing is specified in Figure~\ref{fig:typing:expression}.
The typing rules \tvar and \tparen are straightforward.
The \tprod rule says that the type of an outer product is formed by concatenating the types of the operands.
A transposition swaps two dimensions in a tensor type, cf.~rule \ttrans.
For a tensor contraction to be well-typed, the contracted dimensions must have the same extent, which is ensured by the premise $d_m=d_n$ of rule \tcontr.
The contracted dimensions are then dropped from the resulting type, %
which is indicated by the hats in $(d_1, \dots, \widehat{d_m}, \dots, \widehat{d_n}, \dots, d_k)$.

Rule \telem allows element-wise application of the usual arithmetic operators ($+$, $-$, $*$, and $/$) if the types of the operands match.
Rule \tsmul further allows multiplication by a scalar from the left, and rule \tsdiv analogously allows division by a scalar.
\out{%
	Neither of these scalar operations change the type of their tensor operands.
}

Note that the formation of programs can fail in essentially five ways:
\begin{enumerate}[{(1)}]
	\item
	re-declaration of a variable, violating the premises of rule \svar;
	\item
	assignment to an undeclared variable, violating the premises of \okstmt;
	\item
	mismatch of types in an assignment, violating the premises of rule \okstmt;
	\item
	use of an undeclared variable in an expression, violating the premises of \tvar;
	\item
	mismatch of types in an expression, violating the premises of rules \tcontr, \telem, \tsmul, or \tsdiv.
\end{enumerate}
\begin{figure*}[h!]
	\begin{minipage}{\textwidth}
		\begin{gather*}
		\infer[\dempty]
		{\emptyset \vdash_{d} \bot}
		{} \\[2mm]
		\infer[\dvar]
		{\decl^*\:\:
			\textbf{var}\: \x\: \textbf{:}\: \textbf{[}d_1, \dots, d_k\textbf{]}
			\vdash_{d} \mu\{\forall_{\i\le\t}\:\x_{\i} \mapsto x_{\i}\}}
		{\decl^* \vdash_{d} \mu \quad
			\t = (d_1, \dots, d_k) \quad
			\forall_{\i\le\t}\: x_{\i}\in\VV\cup\{\bullet\}}
		\end{gather*}
		\vspace{-6mm}
		\captionof{figure}{Dynamic store formation $\vdash_{d}$.}
		\label{fig:dynamic:store}
	\end{minipage}%
	\\[7mm]
	\begin{minipage}{\textwidth}
		\begin{gather*}
		\infer[\evstmt]
		{\langle\mu, \x\: \textbf{=}\: \e\rangle
			\rightarrow_{\Gamma}
			\langle \mu\{\forall_{\i\le\t}\:
						\x_{\i} \mapsto r_{\i}\},
					\emptyset\rangle}
		{\begin{array}{c}
			\x\in\dom{\Gamma} \quad
			\t = \Gamma(\x) \\[1mm]
			\forall_{\i\le\t} \:
			\x_{\i}\!\in\!\dom{\mu} \:\wedge\:
			r_{\i}\!=\!\eval{\Gamma\!,\mu}{\e_{\i}}
			\end{array}
		} \\[3mm]
		\infer[\evseq]
		{\langle\mu, \stmt^*\:\stmt\rangle \rightarrow_{\Gamma} \langle\mu'', \emptyset\rangle}
		{\langle\mu, \stmt^*\rangle \rightarrow_{\Gamma} \langle\mu', \emptyset\rangle \quad
			\langle\mu', \stmt\rangle \rightarrow_{\Gamma} \langle\mu'', \emptyset\rangle} \\[2mm]
		\infer[\evprog]
		{\decl^*\:\stmt^* \Downarrow \mu'}
		{\decl^* \vdash_{s} \Gamma \quad
			\decl^* \vdash_{d} \mu \quad
			\langle\mu, \stmt^*\rangle {\rightarrow}_{\Gamma}^{0,1}
			\langle\mu', \emptyset\rangle}
		\end{gather*}
		\vspace{-6mm}
		\captionof{figure}{Statement and program evaluation.}
		\label{fig:evaluation:statement}
	\end{minipage}%
	\\[7mm]
	\begin{minipage}{\textwidth}
		\begin{align}
		&\eval{\Gamma\!,\mu}{\x_{i_1\dots i_k}} =
		\mu(\x_{i_1 \dots i_k})
		\label{eq:eval:var} \\[2mm]
		&\eval{\Gamma\!,\mu}{\textbf{(}\e\,\textbf{)}_{i_1\dots i_k}} = 
		\eval{\Gamma\!,\mu}{\e_{i_1\dots i_k}}
		\label{eq:eval:paren} \\[2mm]
		&\eval{\Gamma\!,\mu}{
			\ez\,\textbf{\#}
			\eo\,_{i_{01}\dots i_{0k} i_{11} \dots i_{il}}
		} = 
		\eval{\Gamma\!,\mu}{\ez_{i_{01} \dots i_{0k}}} \cdot \eval{\Gamma\!,\mu}{\eo_{i_{11} \dots i_{il}}}
		\label{eq:eval:prod} \\[2mm]
		&\eval{\Gamma\!,\mu}{
			\e\,\textbf{\^{}}\,\textbf{[}m\, n\textbf{]}_{i_1 \dots i_m \dots i_n \dots i_k}
		} =
		\eval{\Gamma\!,\mu}{\e_{i_1 \dots i_n \dots i_m \dots i_k}}
		\label{eq:eval:trans} \\[2mm]
		&\eval{\Gamma\!,\mu}{
			\e\,\textbf{.}\,\textbf{[}m\, n\textbf{]}
			_{i_1\dots\widehat{i_m}\dots \widehat{i_n}\dots i_k}
		} =
		\sum_{l=1}^{\pi_{m}(\t)}
		\eval{\Gamma\!,\mu}{
			\e
			_{i_1\dots i_{m-1} l\,i_{m+1}\dots i_{n-1} l\,i_{n+1}\dots i_k}
		} \,, \nonumber \\[-3mm]
		&\hspace{7.2cm} \textrm{where}\:\: \Gamma \vdash \e : \t
		\label{eq:eval:contr} \\[4mm]
		&\eval{\Gamma\!,\mu}{
			\ez\, \odot \eo\,_{i_1 \dots i_k}
		} =
		\left\{
		\begin{array}{l}
		\eval{\Gamma\!,\mu}{\ez\,} \cdot \eval{\Gamma\!,\mu}{\eo_{i_1\dots i_k}} \,, \\[1mm]
		\hspace{3cm} \textrm{if}\:\: \Gamma \vdash \ez : () \,\:\textrm{and}\:\: \odot\equiv* \\[2mm]
		\eval{\Gamma\!,\mu}{\ez_{i_1\dots i_k}} /  \eval{\Gamma\!,\mu}{\eo\,} \,, \\[1mm]
		\hspace{3cm} \textrm{if}\:\: \Gamma \vdash \eo : () \,\:\textrm{and}\:\: \odot\equiv / \\[2mm]			\eval{\Gamma\!,\mu}{\ez_{i_1\dots i_k}}\odot \eval{\Gamma\!,\mu}{\eo_{i_1\dots i_k}} \,, \\[1mm]
		\hspace{3cm} \textrm{otherwise}
		\end{array}
		\right.
		\label{eq:eval:elem}
		\end{align}
		\captionof{figure}{Expression evaluation.}
		\label{fig:evaluation:expression}
	\end{minipage}
\end{figure*}%
\subsubsection{Evaluation of model language programs}
\label{sec:evaluation}
The model language is intended to capture only part of an ambient numerical application, namely the part that deals with manipulating tensors.
Hence, when a model language program is started, it is assumed that memory has been allocated for variables by the ambient application.
The initial store is thus populated with values for all variables declared in the model language program.
Values may be undefined ($\bullet$), which models the situation where the ambient application has allocated memory but has not initialized it.
Uninitialized variables may of course be assigned to during the execution of a model language program.
As already explained, the state of the store $\mu$ after executing a model language program is considered the result of the program.%
\footnote{%
	The interface between the ambient application and a program written in the model language is further discussed in Section \ref{sec:limitations}.
}

The rules in Figure~\ref{fig:dynamic:store} specify the state of the store when execution of a model language program begins.
Rule \dvar captures the requirement that memory has been allocated by the ambient application.
The values $x_{\i}$ (written in italics) are assumed to have been fixed by the ambient application, including the possibility of uninitialized values, which is why the $x_{\i}$ are chosen from $\VV\cup\{\bullet\}$.
Note that since re-declaration of variables is not allowed, each application of \dvar \emph{augments} the domain of $\mu$ with $\x_{\i}$ for all multi-indices $\i\le\t$.

The evaluation of statements and programs is specified by the inference rules in Figure~\ref{fig:evaluation:statement}.
Rule \evprog formalizes the fact that the final state of the dynamic store is considered the result of program evaluation: %
the relation $\Downarrow$ relates a program, which is a purely syntactic entity, to a dynamic store.
The way in which assignment statements manipulate the store is defined by the relation $\rightarrow_{\Gamma}$.
It should be stressed that rule \evstmt \emph{updates} the store $\mu$ with new values for all $\x_{\i}$ with $\i\le\t$.

It may seem surprising that the evaluation relation $\rightarrow_{\Gamma}$, which governs the dynamic behavior of programs, depends explicitly on the static context $\Gamma$.
This can be understood by considering how one would implement the rule \evstmt.
Due to the universal quantifier in the conclusion of the rule, a new value is assigned to each element of the tensor $\x$.
In practice, this would be done inside a loop nest that iterates over all values that can be taken by the indices of $\x$, i.e.~a loop nest that iterates over all multi-indices $\i\le\t$.
Since the type $\t$ is statically known, the machine code produced by a model language compiler would contain constant loop bounds corresponding to the entries in the tuple $\t$, %
and the loop bounds, of course, determine the dynamic behavior of program execution.
Note further that the evaluation of expressions in the premise of rule \evstmt also depends on the static context, as is explained below.

The notation $\rightarrow_{\Gamma}^{0,1}$ in the premise of rule \evprog means that the relation $\rightarrow_{\Gamma}$ is applied zero or one times.
Hence, a program that does not contain any statements evaluates to the initial store, %
i.e.~$\langle\mu, \emptyset\rangle = \langle\mu, \emptyset\rangle$ implies $\decl^* \Downarrow \mu$.

Figure~\ref{fig:evaluation:expression} gives denotational semantics for the evaluation of expressions.
The evaluation function \FNeval depends on both the dynamic store $\mu$ and the static context $\Gamma$.
Equation~\eqref{eq:eval:var} explains why the store is required for evaluation: %
the values bound to variables must be looked up in the store.
The static context $\Gamma$ appears in Equations~\eqref{eq:eval:contr} and~\eqref{eq:eval:elem}, both of which require type information for correct evaluation of expressions.
Once again, a model language compiler would encode this type information in the machine code it produces:
in implementing Equation~\eqref{eq:eval:contr}, $\pi_{m}(\t)$ would become a loop bound;
and the actual value of the operator variable $\odot$ in Equation~\eqref{eq:eval:elem} would be turned into an appropriate machine instruction.
In Equation~\eqref{eq:eval:elem}, and henceforth, we use $\equiv$ to denote an identity between syntactic entities.

Program evaluation can get stuck if identifiers are not in the domain of the store $\mu$ %
or if the function \FNeval is not defined for a given argument.
The following sections show, however, that this cannot happen in well-formed programs.
(``Well-type[d] programs cannot go wrong.''~\cite{Milner:1978})
\subsubsection{Properties of well-formed programs}
\label{sec:properties}
In this section, we establish simple properties of well-formed programs, culminating in Corollary~\ref{cor:eval:well-defined}, which states that the evaluation function \FNeval is well-defined for expressions that appear in well-formed programs.
This already goes a long way towards establishing progress and safety of well-formed programs, which will be done in Section \ref{sec:progress:safety}.

The rules for forming the context $\Gamma$ (in Figure~\ref{fig:typing:context}) and the rules for forming the store $\mu$ (in Figure~\ref{fig:dynamic:store}) have very similar structure.
Therefore, a given context $\Gamma$ uniquely determines the corresponding store $\mu$.
This is formally established by the following Lemma~\ref{lem:dynamic:store}, which assumes that the values $x_{\i}$ in the inference rule \dvar are uniquely determined by the ambient application (cf.~the discussion at the beginning of Section~\ref{sec:evaluation}).
\begin{lemma}[Unique dynamic store]
\label{lem:dynamic:store}
Let $\decl^* \vdash_{s} \Gamma$. %
There exists a unique store $\mu$
such that
\begin{enumerate}[{(i)}]
	\item
	\label{lem:dynamic:store:i}
	$\decl^* \vdash_{d} \mu$, and
	\item 
	\label{lem:dynamic:store:ii}
	for all $\x \in\dom{\Gamma}$, 
	with $\t = \Gamma(\x)$,
	the following holds:
	$\forall_{\i\le\t}\:\x_{\i}\in\dom{\mu}$.
\end{enumerate}
\begin{proof}
By induction on the length of $\decl^*$, which we denote as $N$. \\
~\\[-3mm]
$N=0$: %
The rules \sempty and \dempty imply that $\Gamma=\bot$ and $\mu=\bot$ respectively.
Therefore, $\dom{\Gamma} = \dom{\mu} = \emptyset$ %
and part~(\ref{lem:dynamic:store:ii}) holds trivially. \\
~\\[-3mm]
$N+1$: Let %
\begin{align*}
	\decl^{N+1} \equiv
	\decl^N \:\: \var\,\y :
		\textbf{[}d_1,\dots,d_k\textbf{]}
\end{align*}
such that $\decl^N \vdash_{s} \Gamma_{\!N}$, %
and let $\mu_N$ be the corresponding store that exists by the induction hypothesis, i.e.~$\decl^N \vdash_{d} \mu_N$.
Then, letting $\t = (d_1,\dots,d_k)$, $\Gamma = \Gamma_{\!N}\{\y \mapsto \t\}$.
By the \dvar rule, %
\begin{align*}
	\decl^{N+1} \vdash_{d} \mu, \quad
	\textrm{where}\:\: \mu = \mu_N\{\forall_{\i\le\t}\, \y_{\i} \mapsto y_{\i}\}.
\end{align*}
Clearly, $\dom{\Gamma} = \dom{\Gamma_{\!N}} \cup \{\y\}$.
For $\x\in\dom{\Gamma_{\!N}}$, $\x_{\i}\in\dom{\mu}$ holds by the induction hypothesis since $\dom{\mu_N}\subseteq\dom{\mu}$.
For $\y$, part~(\ref{lem:dynamic:store:ii}) of the lemma holds by construction of $\mu$.
\end{proof}
\end{lemma}
\begin{definition}
\label{def:unique:store}
If $\decl^* \vdash_{s} \Gamma$, then the corresponding dynamic store that exists by Lemma \ref{lem:dynamic:store} is denoted as $\mu(\Gamma)$.
\end{definition}
In the light of Lemma \ref{lem:dynamic:store}, the evaluation rule \evprog can be simplified.
Instead of relying on the relation $\vdash_{d}$ to form the initial store, one can directly use the store $\mu(\Gamma)$:
\begin{gather*}
\infer[\evprogp]
	{\decl^*\:\stmt^* \Downarrow \mu'}
	{\decl^* \vdash_{s} \Gamma \quad
	 \mu = \mu(\Gamma) \quad
	 \langle\mu, \stmt^*\rangle {\rightarrow}_{\Gamma}^{0,1}
	 \langle\mu', \emptyset\rangle}
\end{gather*}%
\out{
Second, the condition $\x_{\i}\in\dom{\mu}$ in the rule \evstmt can be phrased in terms of the domain of $\mu(\Gamma)$,
\begin{gather*}
\infer[\evstmtp]
	{\langle\mu', \x\: \textbf{=}\: \e\rangle
	 \rightarrow_{\Gamma}
	 \langle\mu'\{\forall_{\i\le\t}\:\x_{i_1\dots i_k}
	 \mapsto
	 r_{i_1\dots i_k}\}, \emptyset\rangle}
	{\begin{array}{c}
		\x\in\dom{\Gamma} \quad
	 	\t = \Gamma(\x) \\[1mm]
	 	\dom{\mu(\Gamma)}\subseteq\dom{\mu'} \quad
		 \forall_{\i\le\t} \:
		 r_{\i} \!=\!
	 	\eval{\Gamma\!,\mu(\Gamma)}{\e_{\i}}
	 \end{array}
 	}
 	\: .
\end{gather*}%
}

We could now proceed to proving that $\eval{\Gamma\!,\mu(\Gamma)}{\cdot}$ is well-defined.
However, the proof does not rely on the details of the store $\mu(\Gamma)$, and only the domain of the store matters.
Therefore, the following slightly more general lemma holds.
\begin{lemma}
\label{lem:eval:well-defined}
Assume $\decl^* \vdash_{s} \Gamma$.
The following hold:
\begin{enumerate}[{(i)}]
	\item \label{lem:eval:well-defined:i}
	Let $\mu'$ be a store such that $\dom{\mu(\Gamma)}\subseteq\dom{\mu'}$.
	If $\Gamma \vdash \e : \t$,
	then $\eval{\Gamma\!,\mu'}{\e_{\i}}$ is well-defined for all $\i\le\t$.
	(Note that $\eval{\Gamma\!,\mu'}{\e_{\i}} \in \VV\cup\{\bullet\}$.)
	\item \label{lem:eval:well-defined:ii}
	If $\Gamma \vdash \stmt^*\:\ok$,
	then $\x\in\dom{\Gamma}$ for all statements $\x\:\textbf{=}\:\e$ in $\stmt^*$.
	\item \label{lem:eval:well-defined:iii}
	If $\Gamma \vdash \stmt^*\:\ok$,
	then $\Gamma \vdash \e : \t$, where $\t=\Gamma(\x)$,
	for all statements $\x\:\textbf{=}\:\e$ that occur in $\stmt^*$.
\end{enumerate}
\begin{proof}
Part (\ref{lem:eval:well-defined:i}) is proved by induction on the structure of the expression $\e$. \\
~\\[-3mm]
Case $\e \equiv \x$: %
Since $\Gamma \vdash \e : \t$, inverting rule \tvar implies that %
$\x\in\dom{\Gamma}$ and $\t = \Gamma(\x)$.
By Lemma~\ref{lem:dynamic:store}, $\x_{\i}\in\dom{\mu(\Gamma)} \subseteq \dom{\mu'}$ for all $\i\le\t$.
It thus follows from Equation~\eqref{eq:eval:var} in Figure~\ref{fig:evaluation:expression} that $\eval{\Gamma\!,\mu'}{\e_{\i}}$ is well-defined. \\
~\\[-3mm]
Case $e \equiv \ez\,\textbf{\#}\eo$: %
By inverting rule \tprod, $\Gamma \vdash \ez : \t_0$ and $\Gamma \vdash \eo : \t_1$.
Thus, by the induction hypothesis, $\eval{\Gamma\!,\mu'}{\ez_{\i}}$ and $\eval{\Gamma\!,\mu'}{\eo_{\j}}$ are well-defined for $\i\le\t_0$ and $\j\le\t_1$.
It follows from Equation~\eqref{eq:eval:prod} in Figure~\ref{fig:evaluation:expression} that $\eval{\Gamma\!,\mu'}{\e_{\k}}$ is well-defined %
for $\k\le(d_{01},\dots,d_{0l_0},d_{11},\dots,d_{1l_1})$, where %
$(d_{01},\dots,d_{0l_0}) = \t_0$ and $(d_{11},\dots,d_{1l_1}) = \t_1$. \\
~\\[-3mm]
The remaining cases are handled analogously.
The division operator that occurs in Equation~\eqref{eq:eval:elem} in Figure~\ref{fig:evaluation:expression} poses no difficulty since \FNeval is allowed to take the value $\bullet$ (cf.~Definition~\ref{def:extended:arithmetic}).
The only interesting case is that of contraction. \\
~\\[-3mm]
Case $\e \equiv \ez\,\textbf{.}\,\textbf{[}m\,n\textbf{]}$: %
By inverting rule \tcontr, %
$\Gamma \vdash \ez : \t_0$, %
$\t_0 = (d_0,\dots,d_{k})$, %
and $d_m = d_n$.
Hence, $\eval{\Gamma\!,\mu'}{\ez_{\j}}$ is well-defined for all $\j\le\t_0$ by the induction hypothesis.
On the right-hand side of Equation~\eqref{eq:eval:contr} in Figure~\ref{fig:evaluation:expression}, the summation index $l$ takes values from $\{1, \dots, d_{m}\}$.
It is obvious from this that $\eval{\Gamma\!,\mu'}{\e_{\i}}$ is well-defined %
for $\i\le(d_1,\dots,\widehat{d_m},\dots,\widehat{d_n},\dots,d_k)$. \\
~\\[-3mm]
Parts~(\ref{lem:eval:well-defined:ii}) and~(\ref{lem:eval:well-defined:iii}) are proved by induction on the length of $\stmt^*$.
The proofs use the inversions of rules \okseq and \okstmt.
\end{proof}
\end{lemma}
\newpage
As an immediate consequence of Lemma~\ref{lem:eval:well-defined}, we obtain the result that \FNeval is well-defined on all expressions that occur in a well-formed program.
\begin{corollary}[Well-defined evaluation]
\label{cor:eval:well-defined}
Let $\decl^*\:\stmt^*\:\:\ok$, 
let $\Gamma$ be the static context for which $\decl^*\vdash_{s} \Gamma$,
and let $\mu'$ be a store such that $\dom{\mu(\Gamma)} \subseteq \dom{\mu'}$.
Then, for all $\x\:\textbf{=}\:\e$ that occur in $\stmt^*$, the following hold:
\begin{enumerate}[{(i)}]
	\item
	$\Gamma \vdash \e : \t$, where $\t = \Gamma(\x)$,
	\item
	$\eval{\Gamma\!,\mu'}{\e_{\i}}$ is well-defined for all $\i\le\t$.
\end{enumerate}
\end{corollary}
Corollary~\ref{cor:eval:well-defined} implicitly says that in a well-formed program there are no out-of-bounds read accesses to the store.
To see this, note that a read access occurs precisely when Equation~\eqref{eq:eval:var} from Figure~\ref{fig:evaluation:expression} is used in evaluating an expression.
Since Corollary~\ref{cor:eval:well-defined} guarantees that evaluation is well-defined, the store on the right-hand side of Equation~\eqref{eq:eval:var} can never be accessed outside of its domain.

The absence of out-of-bounds write accesses will follow from the key result of the next section, namely that well-formed programs can be fully evaluated.
\subsection{Progress and safety}
\label{sec:progress:safety}
We now establish the central result that the execution of well-formed programs %
cannot get stuck.
This means that for a well-formed program there always exists a store $\mu$ such that the program and $\mu$ are related by $\Downarrow$.
Thus, evaluation of programs in the model language according to Figure~\ref{fig:evaluation:statement} carries the flavor of big-step semantics.
However, for a well-formed program $P$, the rules \okprog, \okseq and \evprog, \evseq also guarantee that a program $P'$ that is formed of the same declarations as $P$ but contains only a subset of the statements in $P$ can also be fully evaluated.%
\footnote{%
	Since the programs $P$ and $P'$ contain the same declarations, the same static context $\Gamma$ is used in establishing that $P$ and $P'$ are well-formed.
	Hence, Corollary~\ref{cor:eval:well-defined} yields that evaluation of expressions is well-defined for both $P$ and $P'$.
}
This observation lets program evaluation appear closer to small-step semantics.
We therefore use the terms program evaluation and \emph{progress} interchangeably.
\begin{theorem}[Program evaluation]
\label{thm:program:evaluation}
Let $\decl^*\:\stmt^*\:\:\ok$, and let $\Gamma$ be the static context for which $\decl^*\vdash_{s} \Gamma$.
Then, there exists a unique store $\mu'$ such that $\decl^*\:\stmt^* \Downarrow \mu'$ and $\dom{\mu'} = \dom{\mu(\Gamma)}$.
\begin{proof}
By induction on the length of $\stmt^*$, which we denote as $N$. \\
~\\[-3mm]
$N=0$: %
Since $\langle\mu(\Gamma), \emptyset\rangle \rightarrow^0_{\Gamma} \langle\mu(\Gamma),\emptyset\rangle$, rule \evprog can be applied and $\mu'=\mu(\Gamma)$.
(Note that $\rightarrow^0_{\Gamma}$ is the same as equality $=$.) \\
~\\[-3mm]
$N+1$: %
Assume that $\decl^*\:\stmt^N\:\x\:\textbf{=}\:\e\:\:\ok$.
By the induction hypothesis,
\begin{align*}
	&\decl^* \stmt^N \Downarrow \mu_N, \\
	&\dom{\mu_N} = \dom{\mu(\Gamma)}.
\end{align*}
By Lemma~\ref{lem:eval:well-defined}(\ref{lem:eval:well-defined:ii}), %
$\x\in\dom{\Gamma}$.
Thus, let $\t = \Gamma(\x)$.
It follows from Lemma~\ref{lem:dynamic:store} %
that $\x_{\i} \in\dom{\mu(\Gamma)} = \dom{\mu_N}$ for all $\i\le\t$.
Moreover, Corollary~\ref{cor:eval:well-defined} implies %
that $\eval{\Gamma\!,\mu_N}{\e_{\i}}$ is well-defined for $\i\le\t$.
Therefore, we can apply rule \evstmt to obtain
\begin{align*}
	&\langle\mu_N, \x\:\textbf{=}\:\e\rangle \rightarrow_{\Gamma}
		\langle\mu',\emptyset \rangle , \quad
		\textrm{where} \:\:	\mu' = \mu_N\{\forall_{\i\le\t}\:
											\x_{\i}
											\mapsto r_{\i}\} , \\
	& \hspace{3.9cm} \textrm{and} \:\: r_{\i} = \eval{\Gamma\!,\mu_N}{\e_{\i}}.					
\end{align*}
Since $\x_{\i} \in \dom{\mu_N}$ for all $\i\le\t$, %
also $\dom{\mu'} = \dom{\mu_N} = \dom{\mu(\Gamma)}$.
Furthermore, rule \evseq can now be applied to deduce %
\begin{align*}
	\langle\mu(\Gamma), \stmt^N\:\x\:\textbf{=}\:\e \rangle
		\rightarrow_{\Gamma} \langle\mu',\emptyset\rangle
	.
\end{align*}
From this, an application of \evprog yields %
$
	\decl^*\:\stmt^N\:\x\:\textbf{=}\:\e \:\Downarrow\: \mu'
$.
Note that $\mu'$ is unique by construction.
\end{proof}
\end{theorem}
The fact that well-formed programs can be fully evaluated almost immediately implies that programs are safe, %
in the sense that there are no out-of-bounds accesses.
This is because out-of-bounds accesses would evaluate the map $\mu$ outside of its domain, which in turn would cause program evaluation to become stuck.
The following corollary formalizes the absence of out-of-bounds accesses.
\begin{corollary}[Safety]
\label{cor:safety}
Let $\decl^*\:\stmt^*\:\:\ok$.
There are no out-of-bounds accesses to the store during the evaluation of $\stmt^*$.
\begin{proof}
For read accesses to the store, the claim follows already from Corollary~\ref{cor:eval:well-defined}, as explained at the end of Section~\ref{sec:properties}.
For write accesses, the claim is a consequence of Theorem~\ref{thm:program:evaluation} since complete evaluation of $\stmt^*$ requires successive applications of the rule \evstmt.
The premises of this rule ensure that the store $\mu$ is modified only in places that are already in $\dom{\mu}$.
\end{proof}
\end{corollary}
It is important to stress that the progress and safety results establish the sanity of the defined model language.
Theorem~\ref{thm:program:evaluation} implies that the evaluation scheme defined in Figures~\ref{fig:dynamic:store}--\ref{fig:evaluation:expression} is a sensible one: %
it is guaranteed that valid programs can be evaluated without getting stuck.
Corollary~\ref{cor:safety} is of particular importance since the model language provides an index-free way of manipulating tensor expressions.
Hence, accesses to elements of tensors, via multi-indices, are only inserted by an implementation of the model language, e.g.~by a model language compiler.
Corollary~\ref{cor:safety} guarantees that an implementation of the evaluation scheme from Figures~\ref{fig:dynamic:store}--\ref{fig:evaluation:expression} introduces no memory accesses that are out-of-bounds.
\subsection{Applicability and limitations of the model language}
\label{sec:limitations}
The model language is not intended for specifying entire numerical applications.
Instead, it focuses on expressing those parts of an application that specifically manipulate tensors using the previously defined operations, e.g.~contraction, transposition etc.
From a practical perspective, the model language should be thought of as a DSL or intermediate language that facilitates, for example, the generation of efficient code for performance-critical parts of an application by exploiting domain knowledge, analogous to~\cite{Kirby:2006,Luporini:2015:COFFEE,Springer:2017:TTC,Rink:2018:RWDSL}.
This means that there must be an interface between programs written in the model language and the ambient application, usually written in a general-purpose programming language.
The variables declared in the model language collectively constitute this interface.

According to Section~\ref{sec:evaluation}, it is assumed that the ambient application has allocated memory for the variables that are declared in the model language.
Thus, variables can be used to pass values between model language programs and the ambient application.
As already mentioned in Section~\ref{sec:evaluation}, the ambient application may (but need not) initialize variables before a model language program is executed; %
and the values that are assigned to variables in the model language can subsequently be accessed by the ambient application.

The model language supports only a single, global scope since it is only intended to model parts of a numerical application that are limited in size.
Moreover, allowing scopes to be nested would unnecessarily complicate the simple interface between model language programs and ambient applications since it is not obvious if and how variables declared in nested scopes should correspond to variables in the ambient application.

The control flow in numerical applications is typically highly structured: %
loop nests iterate over the possible values that can be taken by the indices of tensors.
The loop bounds that appear in these loop nests are precisely the tensor dimensions, which are typically specified as constants at compile time since this aids compilers in generating more efficient machine code, cf.~\cite{Luporini:2015:COFFEE,Rink:2018:RWDSL}.
Hence, the model language's requirement that tensor dimensions be specified as (compile-time) constants is not a practically relevant limitation.%
\footnote{%
	If this is seen to be a limitation, it can be overcome by deferring the compilation of model language programs until the ambient numerical application is run and tensor dimensions are known, cf.~\cite{Rink:2018:RWDSL}.
}
Of course, general control flow, beyond the structured control flow that is implicit in tensor evaluation and assignment, cannot be expressed in the model language since there are no syntactic constructs for branching or loops.
Once again, this is completely aligned with the separation between ambient applications and model language programs.
Control flow and any complex logic beyond tensor expressions must be specified in the ambient application, which is assumed to be written in a general-purpose programming language.
\section{Padding and alignment}
\label{sec:padding}
In this section we study an alternative implementation of the model language, namely we introduce padding in the memory layout of tensors.
Here, padding means that the memory allocated for any dimension of a tensor is rounded up to a multiple of a fixed value $\M\in\NN$.
Padding can positively affect performance, for example when the architecture of the targeted machine supports vector or SIMD instructions, cf.~\cite{Luporini:2015:COFFEE}.
By padding, the dimensions of tensors can be made multiples of the machine's vector length or SIMD width.
Vectorizing compilers can then produce more efficient machine code.%
\footnote{%
	This is due to the absence of scalar loop iterations that result from vectorization if tensor dimensions are not multiples of the vector length or SIMD width.
}
Another advantage of padding arises in conjunction with alignment: %
aligned memory accesses are usually faster than non-aligned ones, and %
if the memory allocated for a tensor is aligned, then padding can ensure that nested (sub-)tensors are also aligned, cf.~\cite{Luporini:2015:COFFEE}.

The advantage of having formalized the semantics of the model language in Section~\ref{sec:definition} is that an implementation with padding can be proven correct.
This, of course, also requires a formal specification of the implementation with padding, %
which relies on the following Definitions~\ref{def:length:rounding} and~\ref{def:controlled:divison}.

Henceforth, we refer to the constant $\M$ as the \emph{vector length}.
This terminology derives from choosing the value of $\M$ based on the target machine's vector length or SIMD width.
The following definition introduces notation for rounding up natural numbers to the nearest multiple of $\M$.
\begin{definition}[Vector length, rounding]
\label{def:length:rounding}
Let $\M \in \NN$ be the \emph{vector length}.
\begin{enumerate}[{(i)}]
	\item
	For $d\in\NN_0$, %
	let $\up{d}_{\M} = \min\{n\cdot \M : n\in\NN_0 , n\cdot \M \ge d\}$.
	\item
	For $k\in\NN$ and $\i = (i_1,\dots,i_k)\in\NN^k$, %
	let $\up{\i}_M = (\up{i_1}_{\M},\dots,\up{i_k}_{\M})$.
	\item
	For $()\in\NN^0$, let $\up{()}_{\M} = ()$. \\
	(Recall that, by convention, $\NN^0=\{()\}$.)
	\item
	For a static context $\Gamma$, let $\up{\Gamma}_{\M}$ be defined point-wise, i.e.
	\begin{align*}
		\up{\Gamma}_{\M}(\x) = \up{\Gamma(\x)}_{\M},
	\end{align*}		
	for all $\x\in\dom{\Gamma}$.
\end{enumerate}%
~\\[-11mm]
\end{definition}
Subsequently, $\M$ is assumed fixed.
Hence, we omit the subscript on $\up{\cdot}_{\M}$, and simply write $\up{\cdot}$.
Definition~\ref{def:length:rounding} immediately leads to the next lemma, which will be used in the subsequent formal analysis.
\begin{lemma}
\label{lem:padded:expression:typing}
Let $\Gamma$ be a static context.
If $\Gamma \vdash \e : \t$, then $\up{\Gamma} \vdash \e : \up{\t}$.
\begin{proof}
	Straightforward structural induction on the expression $\e$.
	The proof uses the fact that the rules for expression typing in Figure~\ref{fig:typing:expression} are syntax-directed.
\end{proof}
\end{lemma}

The equivalence of the model language implementation from Section~\ref{sec:definition} and the implementation with padding discussed in this section relies on the combination of two facts: %
(1) the operations in Figure~\ref{fig:evaluation:expression} are linear, and %
(2) padded memory can be filled with zeros.
The only place where linearity does not hold is in the denominator of division.
This motivates the introduction of the following \emph{controlled} division operation.%
\footnote{%
	We say that the introduced division operation is \emph{controlled} since dividing zero by zero does not lead to an uncontrolled singularity.%
}
\begin{definition}[Controlled division]
\label{def:controlled:divison}
The \emph{controlled division} operation is defined on $\VV\cup\{\bullet\}$ as in Definition~\ref{def:extended:arithmetic} except that
\begin{enumerate}[{(i)}]
	\item
	$0/0 = 0$,
	\item
	$0/\bullet = 0$.
\end{enumerate}
~\\[-11mm]
\end{definition}
\out{%
	This division operation is \emph{controlled} since $0/0 = 0$ holds generally in the limit where both the numerator and denominator approach zero, but the denominator does so in a more controlled way than the numerator.
}
From now on we assume that division, as it appears, for example, in Figure~\ref{fig:evaluation:expression}, is always controlled division.
The results from previous sections are not affected by this assumption.
\begin{figure*}[h!]
	\begin{minipage}{\textwidth}
		\begin{gather*}
		\infer[\dPADempty]
		{\emptyset \vdash_{pd} \bot}
		{} \\[2mm]
		\infer[\dPADvar]
		{\decl^*\:\:
			\textbf{var}\: \x\: \textbf{:}\: \textbf{[}d_1, \dots, d_k\textbf{]}
			\vdash_{pd} \mu\{\forall_{\i\le\up{\t}}\:\x_{\i} \mapsto x_{\i}\}}
		{\begin{array}{c}
			\decl^* \vdash_{pd} \mu \quad
			\t = (d_1, \dots, d_k) \\[1mm]
			\forall_{\i\le\t}\: x_{\i}\in\VV\cup\{\bullet\} \quad
			\forall_{\i\in(\t,\up{\t})}\: x_{\i} = 0
			\end{array}
		}
		\end{gather*}
		\vspace{-6mm}
		\captionof{figure}{%
			Formation of a padded dynamic store $\vdash_{pd}$.%
		}%
		\label{fig:padding:store}
	\end{minipage}%
	\\[5mm]
	\begin{minipage}{\textwidth}
		\begin{gather*}
		\infer[\evPADstmt]
		{\langle\mu, \x\: \textbf{=}\: \e\rangle
			\leadsto_{\Gamma}
			\langle \mu\{\forall_{\i\le\up{\t}}\:
			\x_{\i} \mapsto r_{\i}\},
			\emptyset\rangle}
		{\begin{array}{c}
			\x\in\dom{\Gamma} \quad
			\t = \Gamma(\x) \\[1mm]
			\forall_{\i\le\up{\t}} \:
			\x_{\i}\!\in\!\dom{\mu} \:\wedge\:
			r_{\i}\!=\!\eval{\up{\Gamma}\!,\mu}{\e_{\i}}
			\end{array}
		} \\[3mm]
		\infer[\evPADseq]
		{\langle\mu, \stmt^*\:\stmt\rangle \leadsto_{\Gamma} \langle\mu'', \emptyset\rangle}
		{\langle\mu, \stmt^*\rangle \leadsto_{\Gamma} \langle\mu', \emptyset\rangle \quad
			\langle\mu', \stmt\rangle \leadsto_{\Gamma} \langle\mu'', \emptyset\rangle} \\[2mm]
		\infer[\evPADprog]
		{\decl^*\:\stmt^* \Downarrow_{p} \mu'}
		{\decl^* \vdash_{s} \Gamma \quad
			\decl^* \vdash_{pd} \mu \quad
			\langle\mu, \stmt^*\rangle {\leadsto}_{\Gamma}^{0,1}
			\langle\mu', \emptyset\rangle}
		\end{gather*}
		\vspace{-6mm}
		\captionof{figure}{Statement and program evaluation with a padded store.}
		\label{fig:padding:statement}
	\end{minipage}%
\end{figure*}%
\subsection{Program evaluation with padding}
\label{sec:padded:evaluation}
The inference rules in Figures~\ref{fig:padding:store} and~\ref{fig:padding:statement} formally define how model language programs are evaluated when the dimensions of tensors are padded with zeros.
As before, it is assumed that an ambient numerical application allocates memory for tensors.
Allocations are further assumed to include padded memory, and rule \dPADvar captures these assumptions.
Note that, also as before, memory may be left uninitialized by the ambient application.
However, memory that has only been allocated for the purpose of padding must be initialized to zero, which is captured by the requirement $x_{\i} = 0$ for $\i\in(\t,\up{\t})$ in rule \dPADvar.
(Recall from Definition~\ref{def:multi:index}(\ref{def:multi:index:interval}) the meaning of ``$\in$'' in this context.)

As in Section~\ref{sec:definition}, we assume that the values $x_{\i}$ in the premises of rule \dPADvar are uniquely determined by the ambient application.
We then obtain the following uniqueness lemma for a well-formed store $\mu$ with padding.
\begin{lemma}[Unique padded store]
\label{lem:padded:store}
Let $\decl^* \vdash_{s} \Gamma$. %
There exists a unique store $\mu$
such that
\begin{enumerate}[{(i)}]
	\item
	$\decl^* \vdash_{pd} \mu$, and
	\item 
	for all $\x \in\dom{\Gamma}$,
	with $\t = \Gamma(\x)$,
	\begin{enumerate}[{a.}]
		\item
		$\forall_{\i\le\up{\t}}\:\x_{\i}\in\dom{\mu}$,
		\item 
		$\forall_{\i\le\t}\:\mu(\x_{\i}) = \mu(\Gamma)(\x_{\i})$,
		\item 
		$\forall_{\i\in(\t,\up{\t})}\:\mu(\x_{\i}) = 0$,
	\end{enumerate}
	where $\mu(\Gamma)$ is the uniquely defined store from Definition~\ref{def:unique:store}.
\end{enumerate}
~\\[-11mm]
\end{lemma}
\begin{proof}
The proof is completely analogous to that of Lemma~\ref{lem:dynamic:store}.	
\end{proof}
As before, we are compelled to afford a special symbol to the uniquely defined padded store.
\begin{definition}
\label{def:padded:store}
If $\decl^* \vdash_{s} \Gamma$, then the corresponding padded store that exists by~Lemma \ref{lem:padded:store} is denoted as $\up{\mu}(\Gamma)$.
\end{definition}

The rules for the evaluation of statements and programs in Figure~\ref{fig:padding:statement} are almost identical to the ones from Figure~\ref{fig:evaluation:statement}.
The key difference is that the rounded type $\up{\t}$ and context $\up{\Gamma}$ appear in rule \evPADstmt.
In Figure~\ref{fig:evaluation:expression}, the function \FNeval was defined for an arbitrary context $\Gamma$ and dynamic store $\mu$.
Using the context $\up{\Gamma}$ on \FNeval in \evPADstmt lets us reuse the definitions from Figure~\ref{fig:evaluation:expression}.
The relevance of using $\eval{\up{\Gamma}\!,\mu}{\cdot}$ in \evPADstmt will become apparent in the proofs of subsequent results.
\subsection{Progress and safety with padding}
\label{sec:padded:progress}
We now establish the progress and safety properties for the model language with padding.
Both are properties of the relation $\Downarrow_{p}$, and we establish these properties by following the same route as for the relation $\Downarrow$ in Sections~\ref{sec:properties} and~\ref{sec:progress:safety}.
Hence, the progress and safety results in the current section, and their proofs, are analogous to Theorem~\ref{thm:program:evaluation} and Corollary~\ref{cor:safety}.
\begin{lemma}
\label{lem:padded:eval:well-defined}
Assume $\decl^* \vdash_{s} \Gamma$.
Let $\mu'$ be a store such that %
\begin{enumerate}[{a.}]
	\item
	\label{lem:padded:eval:well-defined:a}
	$\dom{\up{\mu}(\Gamma)}\subseteq\dom{\mu'}$, and
	\item
	\label{lem:padded:eval:well-defined:b}
	for all $\x\in\dom{\Gamma}$, with $\t=\Gamma(\x)$, we have $\mu'(\x_{\i}) = 0$ if $\i\in(\t,\up{\t})$.
\end{enumerate}
Then, $\Gamma \vdash \e : \t$ implies
\begin{enumerate}[{(i)}]
	\item
	\label{lem:padded:eval:well-defined:i}
	$\eval{\up{\Gamma}\!,\mu'}{\e_{\i}}$ is well-defined for all $\i\le\up{\t}$, and
	\item
	\label{lem:padded:eval:well-defined:ii}
	$\eval{\up{\Gamma}\!,\mu'}{\e_{\i}} = 0$ for all $\i\in(\t,\up{\t})$.
\end{enumerate}
\begin{proof}
We proceed by structural induction on the expression $\e$.
The proof is very similar to the proof of Lemma~\ref{lem:eval:well-defined}, %
part~(\ref{lem:eval:well-defined:i}). \\
~\\[-3mm]
Case $\e \equiv \x$: %
The lemma follows directly from the definition of \FNeval in Equation~\eqref{eq:eval:var}, Figure~\ref{fig:evaluation:expression}.
Note that the requirement~\ref{lem:padded:eval:well-defined:b}.~is needed to establish part~(\ref{lem:padded:eval:well-defined:ii}). \\
~\\[-3mm]
Case $e \equiv \ez\,\textbf{\#}\eo$: %
Part~(\ref{lem:padded:eval:well-defined:i}) follows as in the proof of Lemma~\ref{lem:eval:well-defined}.
Part~(\ref{lem:padded:eval:well-defined:ii}) follows directly from the induction hypothesis since for $\i\in(\t,\up{\t})$ one of the factors on the right-hand side of Equation~\eqref{eq:eval:prod} in Figure~\ref{fig:evaluation:expression} is equal to zero. \\
~\\[-3mm]
Case $e \equiv \ez\,/\eo$: %
The same reasoning applies as in the proof of Lemma~\ref{lem:eval:well-defined}, %
i.e.~$\eval{\up{\Gamma}\!,\mu'}{\e_{\i}}$ is well-defined since it may take the value $\bullet$.
For $\i\in(\t,\up{\t})$, part~(\ref{lem:padded:eval:well-defined:ii}) follows from the property $0/0 = 0$ of controlled division.
When $\up{\Gamma} \vdash \eo : ()$, i.e.~in the case of division by a scalar, the property $0/\bullet = 0$ of controlled division is also required to establish part~(\ref{lem:padded:eval:well-defined:ii}). \\
~\\[-3mm]
Case $\e \equiv \ez\,\textbf{.}\,\textbf{[}m\,n\textbf{]}$: %
By inverting rule \tcontr, %
$\Gamma \vdash \ez : \t_0$, %
$\t_0 = (d_0,\dots,d_{k})$, %
and $d_m = d_n$.
Again, the fact that $\eval{\up{\Gamma}\!,\mu'}{\e_{\i}}$ is well-defined follows as in the proof of Lemma~\ref{lem:eval:well-defined}, but bear in mind that on the right-hand side of Equation~\eqref{eq:eval:contr} in Figure~\ref{fig:evaluation:expression}, the summation index $l$ now takes values from $\{1, \dots, \up{d_{m}}\}$. \\
~\\[-3mm]
The remaining cases are straightforward.
\end{proof}
\end{lemma}
Note that the proof of part~(\ref{lem:padded:eval:well-defined:ii}) of Lemma~\ref{lem:padded:eval:well-defined} relies on the linearity of the arithmetic operations that occur in the definition of \FNeval in Figure~\ref{fig:evaluation:expression} %
(and on controlled division from Definition~\ref{def:controlled:divison}).

The following progress theorem is analogous to Theorem~\ref{thm:program:evaluation}. It also says that the final store is still suitably padded with zeros.
\begin{theorem}[Program evaluation with a padded store]
\label{thm:padded:program:evaluation}
Let $\decl^*\:\stmt^*\:\:\ok$, and let $\Gamma$ be the static context for which $\decl^*\vdash_{s} \Gamma$.
Then, there exists a unique store $\mu'$ such that $\decl^*\:\stmt^* \Downarrow_{p} \mu'$ and $\dom{\mu'} = \dom{\up{\mu}(\Gamma)}$.
Moreover, $\mu'(\x_{\i}) = 0$ %
if $\x\in\dom{\Gamma}$, $\t = \Gamma(\x)$, and $\i\in(\t,\up{\t})$.
\begin{proof}
The proof is similar to that of Theorem~\ref{thm:program:evaluation}. %
Hence, this proof also proceeds by induction on the length of $\stmt^*$, which we denote as $N$. \\
~\\[-3mm]
$N=0$: %
Since $\langle\up{\mu}(\Gamma), \emptyset\rangle \leadsto^0_{\Gamma} \langle\up{\mu}(\Gamma),\emptyset\rangle$, it can be concluded from \evPADprog that $\mu'=\up{\mu}(\Gamma)$.
Lemma~\ref{lem:padded:store} then implies $\mu'(\x_{\i}) = 0$ for $\x\in\dom{\Gamma}$, $\t = \Gamma(\x)$, and $\i\in(\t,\up{\t})$. \\
~\\[-3mm]
$N+1$: %
Assume that $\decl^*\:\stmt^N\:\x\:\textbf{=}\:\e\:\:\ok$.
By the induction hypothesis,
\begin{align*}
	&\decl^*\:\stmt^N \Downarrow_{p} \mu_N, \\
	&\dom{\mu_N} = \dom{\up{\mu}(\Gamma)}.
\end{align*}
By Lemma~\ref{lem:eval:well-defined}(\ref{lem:eval:well-defined:ii}), %
$\x\in\dom{\Gamma}$.
Thus, let $\t = \Gamma(\x)$.
It follows from Lemma~\ref{lem:padded:store} %
that $\x_{\i} \in\dom{\up{\mu}(\Gamma)} = \dom{\mu_N}$ for all $\i\le\up{\t}$.
Also by the induction hypothesis, $\mu_N$ satisfies the requirements of Lemma~\ref{lem:padded:eval:well-defined}.
We can therefore apply rule \evPADstmt to obtain
\begin{align*}
	&\langle\mu_N, \x\:\textbf{=}\:\e\rangle
		\leadsto_{\Gamma} \langle\mu',\emptyset \rangle , \quad
		\textrm{where} \:\:	\mu' = \mu_N\{\forall_{\i\le\up{\t}}\:
			\x_{\i} \mapsto r_{\i}\} , \\
	& \hspace{3.9cm} \textrm{and} \:\: r_{\i} = \eval{\up{\Gamma}\!,\mu_N}{\e_{\i}}.					
\end{align*}
Since $\x_{\i} \in \dom{\mu_N}$ for all $\i\le\up{\t}$, %
also $\dom{\mu'} = \dom{\mu_N} = \dom{\up{\mu}(\Gamma)}$.
Applications of the rules \evPADseq and \evPADprog now yield the desired
\begin{align*}
	\decl^*\:\stmt^N\:\x\:\textbf{=}\:\e \:\Downarrow_{p}\: \mu'
	.
\end{align*}
Finally, let $\y\in\dom{\Gamma}$ and $\t_1 = \Gamma(\y)$.
If $\y \not\equiv \x$, then $\mu'(\y_{\i}) = \mu_N(\y_{\i}) = 0$ %
for $\i\in(\t_1,\up{\t_1})$ by the induction hypothesis.
If $\y \equiv \x$ and $\i\in(\t_1,\up{\t_1})$, then $\mu'(\y_{\i}) = \mu'(\x_{\i}) = r_i = 0$ %
follows by the construction of $\mu'$ and Lemma~\ref{lem:padded:eval:well-defined}(\ref{lem:padded:eval:well-defined:ii}).
\end{proof}
\end{theorem}
As previously, this progress result also implies that model language programs implemented with a padded store are safe, in the sense that there are no out-of-bounds memory accesses.
\begin{corollary}[Safety with a padded store]
\label{cor:padded:safety}
Let $\decl^*\:\stmt^*\:\:\ok$.
There are no out-of-bounds accesses to the padded store during the evaluation of $\stmt^*$ as defined by the rules in Figure~\ref{fig:padding:statement}.
\end{corollary}
\subsection{Simulation with padding}
\label{sec:padded:simulation}
In this section we finally arrive at the central result that program evaluation with a padded dynamic store simulates the original program evaluation from Figure~\ref{fig:evaluation:statement}.
In other words, evaluation with a padded store produces the same values as the original evaluation, and is therefore equivalent to the original evaluation.
As a first step towards this equivalence result, the next lemma relates evaluation in the context $\up{\Gamma}$ to evaluation in the original context $\Gamma$.
\begin{lemma}
\label{lem:padded:evaluation}
Let $\decl^*\vdash_{s} \Gamma$.
Let $\mu_1$ and $\mu_2$ be stores such that
\begin{enumerate}[{a.}]
	\item
	\label{lem:padded:evaluation:a}
	$\dom{\mu(\Gamma)} \subseteq \dom{\mu_1}$ and %
	$\dom{\up{\mu}(\Gamma)} \subseteq \dom{\mu_2}$,
	\item
	\label{lem:padded:evaluation:b}
	for all $x\in\dom{\Gamma}$ with $\t_{\x} = \Gamma(x)$,
	\begin{align*}
		\mu_2(\x_{\i}) = \left\{
			\begin{array}{ll}
				\mu_1(\x_{\i}) \,, & \textrm{if} \:\: \i\le \t_{\x} \\
				0 \,, & \textrm{if} \:\: \i\in(\t_{\x},\up{\t_{\x}})
			\end{array}
			 \right.
			 .
	\end{align*}
\end{enumerate}
If $\Gamma \vdash \e : \t$, %
then $\eval{\Gamma\!,\mu_1}{\e_{\i}} = \eval{\up{\Gamma}\!,\mu_2}{\e_{\i}}$ %
for all $\i\le\t$.
\begin{proof}
The proof proceeds by structural induction on $\e$. \\
~\\[-3mm]
Case $\e \equiv \x$: %
Requirement~\ref{lem:padded:evaluation:b}.~guarantees that the claim of the lemma holds. \\
~\\[-3mm]
Case $\e \equiv \ez \odot \eo$: %
If $\odot\equiv *$ and $\Gamma \vdash \ez : ()$, %
then $\up{\Gamma} \vdash \ez : ()$ by Lemma~\ref{lem:padded:expression:typing}, %
and the claim follows from the induction hypothesis.
Similarly for $\odot\equiv /$ and $\Gamma \vdash \eo : ()$.
The remaining cases follow at once from the induction hypothesis. \\
~\\[-3mm]
Case $\e \equiv \ez\,\textbf{.}\,\textbf{[}m\,n\textbf{]}$: %
By $\Gamma \vdash \e : \t$ and inversion of \tcontr, we have $\Gamma \vdash \ez : \t_0$.
It then follows from Lemma~\ref{lem:padded:expression:typing} that $\up{\Gamma} \vdash \ez : \up{\t_0}$.
Thus, if $\pi_m(\t_0) = d_m$, then $\pi_m(\up{\t_0}) = \up{d_m}$.
Therefore,
\begin{align*}
	&\eval{\up{\Gamma}\!,\mu_2}{
		\ez\,\textbf{.}\,\textbf{[}m\, n\textbf{]}
		_{i_1\dots\widehat{i_m}\dots \widehat{i_n}\dots i_k}
	} \\
	& = 
	\sum_{l=1}^{\up{d_m}}
	\eval{\up{\Gamma}\!,\mu_2}{
		\ez
		_{i_1\dots i_{m-1} l\,i_{m+1}\dots i_{n-1} l\,i_{n+1}\dots i_k}
	} \\
	& =
	\sum_{l=1}^{d_m}
	\eval{\up{\Gamma}\!,\mu_2}{
		\ez
		_{i_1\dots i_{m-1} l\,i_{m+1}\dots i_{n-1} l\,i_{n+1}\dots i_k}
	} ,
\end{align*}
where the second identity holds by Lemma~\ref{lem:padded:eval:well-defined}(\ref{lem:padded:eval:well-defined:ii}),
which, in turn, is applicable by virtue of requirements~\ref{lem:padded:evaluation:a}.,~\ref{lem:padded:evaluation:b}.
The desired result now follows from the induction hypothesis. \\
~\\[-3mm]
The remaining cases are straightforward.
\end{proof}
\end{lemma}

\begin{theorem}[Simulation]
\label{thm:padded:simulation}
Let $\decl^*\:\stmt^*\:\:\ok$, and let $\Gamma$ be the static context for which $\decl^*\vdash_{s} \Gamma$.
Let $\mu_1$ and $\mu_2$ be the uniquely determined stores from Theorems~\ref{thm:program:evaluation} and~\ref{thm:padded:program:evaluation} respectively, %
i.e.~$\decl^*\:\stmt^* \Downarrow \mu_1$ %
and $\decl^*\:\stmt^* \Downarrow_{p} \mu_2$.
Then, if $\x\in\dom{\Gamma}$, $\t = \Gamma(\x)$, %
\begin{align*}
	\mu_2(\x_{\i}) = \left\{\begin{array}{ll}
						\mu_1(\x_{\i}) \,, & \textrm{if} \:\: \i\le \t \\
						0 \,, & \textrm{if} \:\: \i\in(\t,\up{\t})
					 \end{array}\right.
	.
\end{align*}
\begin{proof}
By induction on the length of $\stmt^*$, which we denote as $N$. \\
~\\[-3mm]
$N=0$: %
In this case, $\mu_1=\mu(\Gamma)$ and $\mu_2=\up{\mu}(\Gamma)$, %
by Lemmata~\ref{lem:dynamic:store} and~\ref{lem:padded:store} respectively.
The desired relationship between $\mu_1$ and $\mu_2$ also follows from Lemma~\ref{lem:padded:store}. \\
~\\[-3mm]
$N+1$: %
Let %
\begin{align*}
	\stmt^{N+1} \equiv \stmt^N \:\x\:\textbf{=}\:\e .
\end{align*}
By the induction hypothesis, %
\begin{align*}
	&\decl^*\:\stmt^N \Downarrow \mu_{N,1}   \, , \\
	&\decl^*\:\stmt^N \Downarrow_p \mu_{N,2} \, , \\
	&\mu_{N,2}(\y_{\i}) = \left\{\begin{array}{ll}
							\mu_{N,1}(\y_{\i}) \,, & \textrm{if} \:\: \i\le \t_{\y} \\
							0 \,, & \textrm{if} \:\: \i\in(\t_{\y},\up{\t_{\y}})
	 					  \end{array}\right.
	,
\end{align*}
for $\y\in\dom{\Gamma}$, $\t_{\y} = \Gamma(\y)$.
Moreover, by inversion of \evseq and \evPADseq %
(using the uniqueness of $\mu_{N,1}$ and $\mu_{N,2}$),
\begin{align*}
	&\langle \mu_{N,1}, \x\:\textbf{=}\:\e\rangle \rightarrow_{\Gamma} \langle \mu_1, \emptyset \rangle, \\
	&\langle \mu_{N,2}, \x\:\textbf{=}\:\e\rangle \leadsto_{\Gamma} \langle \mu_2, \emptyset \rangle.
\end{align*}
Hence, by \evstmt and \evPADstmt respectively,
\begin{align*}
	&\mu_1 = \mu_{N,1}\{
				\forall_{\i \le \t_{\x}} \: \x_{\i} \mapsto \eval{\Gamma\!,\mu_{N,1}}{\e_{\i}}
			 \} , \\
	&\mu_2 = \mu_{N,2}\{
				\forall_{\i \le \up{\t_{\x}}} \: \x_{\i} \mapsto \eval{\up{\Gamma}\!,\mu_{N,2}}{\e_{\i}}
			 \} ,
\end{align*}
where $\t_{\x} = \Gamma(\x)$.
For $\y\not\equiv\x$ the claim of the theorem holds by the induction hypothesis.
For $\y\equiv\x$ and $\i \le \t_{\x}$, we have
\begin{align*}
	\mu_2(\x_{\i}) 
	= \eval{\up{\Gamma}\!,\mu_{N,2}}{\e_{\i}}
	= \eval{\Gamma\!,\mu_{N,1}}{\e_{\i}}
	= \mu_1(\x_{\i}) ,
\end{align*}
by Lemma~\ref{lem:padded:evaluation}.
For $\y\equiv\x$ and $\i \in (\t_{\x}, \up{\t_{\x}})$, we have
\begin{align*}
	\mu_2(\x_{\i}) 
	= \eval{\up{\Gamma}\!,\mu_{N,2}}{\e_{\i}}
	= 0 ,
\end{align*}
by Lemma~\ref{lem:padded:eval:well-defined}(\ref{lem:padded:eval:well-defined:ii}).
\end{proof}	
\end{theorem}
Note that for a minimal simulation result, it would suffice that the dynamic stores $\mu_1$ and $\mu_2$ from Theorem~\ref{thm:padded:simulation} agree on the intersection of their domains.
Theorem~\ref{thm:padded:simulation} is slightly stronger in that it also states that the store $\mu_2$ takes the value zero outside of the domain of $\mu_1$.
This additional strength was used in the above proof to yield a strong enough induction hypothesis.
\section{Extensions}
\label{sec:extensions}
As described in Sections~\ref{sec:evaluation} and~\ref{sec:limitations}, a model language program interfaces with the ambient application through the variables declared in the model language program.
Memory for variables is assumed to have been allocated by the ambient application.
This interface is rather crude in that it does not allow one to express which variables are intended to be used for communication with the ambient application or how variables are intended to be used (for input, output, or both).
An easy way to improve this interface is to add qualifiers to variable declarations that indicate how a variable is to be used.
This is analogous to the storage qualifiers \emph{in} and \emph{out} in the OpenGL shading language~\cite{OpenGL:2017}.
The Fortran programming language~\cite{Fortran:2010} allows a similar qualification of procedure arguments with the \emph{intent} attribute.

An \emph{input} qualifier on the declaration of a variable \x in the model language could be used to indicate that \x is intended to communicate data from the ambient application into the model language program.
This would imply that none of the elements of the tensor \x can be undefined when execution of the model language program begins.
Thus, the rules for forming the dynamic store in Figure~\ref{fig:dynamic:store} would have to be modified.
Figure~\ref{fig:dynamic:store:input} gives the modified inference rules.
Note that rule \dinputvar must be applied when a variable declaration with an \emph{input} qualifier is met, %
and this rule disallows the undefined value $\bullet$.
\begin{figure*}[h!]
	\begin{minipage}{\textwidth}
		\begin{gather*}
		\infer[\dempty]
		{\emptyset \vdash_{d} \bot}
		{} \\[2mm]
		\infer[\dvar]
		{\decl^*\:\:
			\textbf{var}\: \x\: \textbf{:}\: \textbf{[}d_1, \dots, d_k\textbf{]}
			\vdash_{d} \mu\{\forall_{\i\le\t}\:\x_{\i} \mapsto x_{\i}\}}
		{\decl^* \vdash_{d} \mu \quad
			\t = (d_1, \dots, d_k) \quad
			\forall_{\i\le\t}\: x_{\i}\in\VV\cup\{\bullet\}} \\[2mm]
		\infer[\dinputvar]
		{\decl^*\:\:
			\textbf{var input}\: \x\: \textbf{:}\: \textbf{[}d_1, \dots, d_k\textbf{]}
			\vdash_{d} \mu\{\forall_{\i\le\t}\:\x_{\i} \mapsto x_{\i}\}}
		{\decl^* \vdash_{d} \mu \quad
			\t = (d_1, \dots, d_k) \quad
			\forall_{\i\le\t}\: x_{\i}\in\VV}
		\end{gather*}
		\vspace{-6mm}
		\captionof{figure}{Dynamic store formation with input qualifier.}
		\label{fig:dynamic:store:input}
	\end{minipage}%
\end{figure*}

Extending the model language with an \emph{input} qualifier also enables a static program analysis that detects uses of uninitialized variables.
By Figure~\ref{fig:dynamic:store:input}, only variables declared as \emph{input} are guaranteed to be initialized when program execution begins.
Thus, if uses of uninitialized variables are to be avoided, all expressions evaluated by a model language program can ultimately only depend on \emph{input} variables.
Due to the absence of explicit control flow in the model language, analyzing the dependencies between expressions is straightforward.

Similarly, an \emph{output} qualifier can be used to enable further program analysis, and even optimization.
The \emph{output} qualifier indicates which variables are \emph{live-out} of a model language program, i.e.~which variables are used to communicate results of the program back to the ambient application, after the model language program has completed.
Any expressions in a model language program that do not affect the value of an \emph{output} variable can be eliminated from the program.
This is a form of an optimization known as \emph{dead code elimination}, cf.~\cite{Muchnick:2000}.
Note that the language implementation in~\cite{Rink:2018:RWDSL} supports \emph{input} and \emph{output} qualifiers.

Generally, the simple structure of the model language should simplify many program analyses and optimizations.
For example, detecting common subexpressions is straightforward, facilitating \emph{common subexpression elimination}.
The absence of explicit control flow turns \emph{reaching definition analysis} into a simple backwards search: %
each use of a variable \x is reached by the last definition of \x, i.e.~the last assignment to \x.
(See again~\cite{Muchnick:2000} for descriptions of the mentioned analyses and optimizations.)

Lastly, it should be noted that many tensor operations are inherently parallel since they independently operate on tensor elements corresponding to different multi-indices.
An advanced implementation of the model language could exploit this parallelism, especially when a GPU architecture is targeted for program execution.
However, a few subtleties must be considered when designing such a parallel implementation.
The rule \evstmt is written such that evaluation of the function \FNeval for different multi-indices \i can proceed in parallel, and the result of \FNeval is assigned to the temporary tensor variable $r$.
Only then is the store $\mu$ updated with the values of the elements of $r$.
A direct implementation of this is rather wasteful since additional memory must be allocated to hold the temporary tensor $r$.
This is likely to lead to less than optimal performance of the memory system, especially caches.
A more sensible implementation would avoid allocating memory for the temporary $r$ and directly commit the result of \FNeval to the store $\mu$.
However, proceeding like this in parallel, i.e.~for different multi-indices \i at the same time, may lead to incorrect results if the variable \x that is assigned to also occurs in the expression \e from the rule \evstmt.
Thus, non-trivial program analysis is required to exploit parallelism correctly and efficiently in evaluating tensor expressions.
\section{Summary}
\label{sec:summary}
In the present report, we have introduced a language for manipulating tensor expressions.
The language serves to model any DSL or tensor library API that supports common tensor operations, e.g.~element-wise arithmetic, contraction, transposition.
Formal semantics of our model language have been specified, and a number of properties have been established formally.

The key properties of our model language definition can be summarized as \emph{progress} and \emph{safety}.
By progress we mean that a well-formed program in the model language can be fully evaluated, without getting stuck.
The property of safety, in the sense that there are no out-of-bounds memory accesses during the evaluation of a model language program, was established almost as a by-product of our progress proof.

We have also formally specified an alternative implementation of the model language in which tensor dimensions are padded up to multiples of the targeted machine's vector length or SIMD width.
In a practical implementation, padding may have a positive impact on performance.
The benefit of having defined formal language semantics is that the original language definition and the implementation with padding can formally be proven equivalent.

A number of extensions of our model language have been discussed, leading to discussions also of (a) the applicability of standard program analyses and optimizations to model language programs and (b) the possibility of exploiting parallelism in tensor operations in a language implementation.
Hence, future work could look into (a) formalizing program analyses and optimizations in the context of the model language and (b) formalizing parallel evaluation semantics for model language programs.
\section*{Acknowledgments}
The author would like to thank %
Sebastian Ertel and %
Andres Goens
for general discussions on formal language semantics.
This work was partially supported by the German Research Council (DFG) through the Cluster of Excellence `Center for Advancing Electronics Dresden' (cfaed).
\newpage
\bibliography{bibliography}
\end{document}